\documentclass[11pt, a4paper]{article}

\usepackage[utf8]{inputenc}
\usepackage[T1]{fontenc}
\usepackage[UKenglish]{babel}
\usepackage{amsmath, amsthm, amssymb}
\usepackage{algorithm, algorithmic}
\usepackage{graphicx}
\usepackage{geometry}
\usepackage[colorlinks=true, allcolors=blue]{hyperref}
\usepackage[nameinlink, capitalize]{cleveref} 
\usepackage{authblk} 
\usepackage[colorlinks=true, allcolors=blue]{hyperref}

\usepackage{orcidlink}

\newtheorem{theorem}{Theorem}
\newtheorem{lemma}[theorem]{Lemma}

\newtheorem{observation}[theorem]{Observation}
\theoremstyle{definition}
\newtheorem{definition}{Definition}
\newtheorem{problem}{Problem}
\newtheorem{remark}{Remark}

\title{Set Parameterized Matching via Multi-Layer Hashing}

\author[1]{Moshe Lewenstein\orcidlink{0000-0002-8272-244X}\thanks{Work on this paper was supported by ISF grant \# 2760/25 awarded by the Israel Science Foundation.}}
\author[1]{Ely Porat \orcidlink{0000-0001-6912-5766}}
\affil[1]{Bar-Ilan University, Israel\\ \texttt{moshe.lewenstein@biu.ac.il, porately@cs.biu.ac.il}}

\date{} 

\begin{document}

\maketitle

\begin{abstract}
We study the \emph{set parameterized matching} problem, a generalization of the classical parameterized matching problem introduced by Baker~\cite{baker1993theory,baker1997parameterized}.
In set parameterized matching, both the pattern and text are sequences where each position contains a \emph{set} of characters rather than a single character.
Two set-strings parameterized match if there exists a bijection between their alphabets that maps one to the other set-wise.
Boussidan~\cite{Boussidan25} introduced this problem for the case of equal-length set-strings.
We present a randomized algorithm running in $O(N + M)$ time with high probability, where $N$ is the text size and $M$ is the pattern size.
Our approach employs a novel three-layer hashing scheme based on Karp-Rabin fingerprinting that addresses the challenges of (1) the size blowup in representations of the problem, (2) set-to-set matching, and (3) the dynamic nature of encodings of text substrings during pattern scanning.
\end{abstract}

\newpage
\section{Introduction}

The parameterized pattern matching problem, introduced by Baker~\cite{baker1992program,baker1993theory,baker1997parameterized}, seeks all occurrences of a pattern $P[1..m]$ in a text $T[1..n]$ where an occurrence is defined not by exact character matches but by consistent character renaming.
Formally, the pattern occurs at position $i$ if there exists a bijection $\pi$ from the pattern alphabet to the text alphabet such that $T[i+j-1] = \pi(P[j])$ for all $j \in [1..m]$.
This problem has found numerous applications in software maintenance~\cite{baker1992program}, plagiarism detection, computational biology~\cite{Shibuya04}, and beyond~\cite{mendivelso2020brief}.

\subsection{Prior Work}

The study of \emph{parameterized matching} was pioneered by Baker~\cite{baker1992program, baker1993theory}, who introduced the concept to identify duplicated code in large software systems.
The core innovation was the ``predecessor string'' (or \emph{prev}) representation, which converts a string into an integer sequence where each character is replaced by the offset to its previous occurrence.
A key observation is that two strings parameterized match if and only if their \emph{prev} transformations are identical.
This insight allowed for efficient matching regardless of specific variable names, provided the usage patterns were consistent.
This discovery enabled $O(n \log \min(m,|\Sigma_P|))$-time algorithms~\cite{amir1994alphabet} and led to the development of parameterized suffix trees~\cite{baker1997parameterized} and suffix arrays~\cite{deguchi2008parameterized}.
Subsequent research has explored the theoretical boundaries and practical variations of the problem:

\begin{itemize}
\item \textbf{Approximate matching:} Significant work has been done on matching with $k$ mismatches~\cite{apostolico2007parameterized, hazay2007approximate}, edit distance~\cite{baker1999parameterized,apostolico2008periodicity}, lcs~\cite{keller2009longest} and $\delta\gamma$-distances~\cite{lee2008delta}.
\item \textbf{Compressed matching:} Solutions have been developed for searching in compressed text~\cite{garg2014fast}, run-length encoded strings~\cite{apostolico2012parameterized}, and other compression schemes~\cite{BealA16,khetan2016efficient}.
\item \textbf{Higher dimensions and Structures:} The problem has been extended beyond linear strings to matching on trees~\cite{amir2009parameterized} and two-dimensional grids~\cite{cole2014two}.
\item \textbf{Succinct data structures:} Recent developments include space-efficient dictionaries~\cite{ganguly2016space,ganguly2017pbwt} and the parameterized Burrows-Wheeler Transform (pBWT)~\cite{ganguly2017pbwt}.
\end{itemize}

For a comprehensive overview of the field's evolution, Lewenstein~\cite{lewenstein2008parameterized} and Mendivelso et al.~\cite{mendivelso2020brief} provide a detailed historical account of these problems and their various applications.

\subsection{Set Parameterized Matching}

While classical parameterized matching handles sequences of single characters, many real-world scenarios involve ambiguity or multiple valid interpretations at a single position.
Boussidan~\cite{Boussidan25} in his PhD thesis (see also~\cite{BourhisBG23}) was concerned about comparing 18th century French plays, where the characters of the play are important, but each scene has multiple characters within.
To address this, Boussidan extended the model of parameterized matching to \emph{set parameterized comparison}, where for two equal-length strings, each position in both contains a \emph{set} of characters, and the question is whether they parameterize match. We call these strings containing sets, {\em set-strings}.
Boussidan's original algorithm relied on bipartite matching. This approach is computationally expensive, requiring $O(m \cdot |\Sigma|^{2.5})$ time for comparison alone.
In this paper, we consider the full pattern matching version of the problem:

\begin{problem}[Set Parameterized Matching]
Given a text $\mathcal{T} = \mathcal{T}[1..n]$ and a pattern $\mathcal{P} = \mathcal{P}[1..m]$ where each $\mathcal{T}[i], \mathcal{P}[j] \subseteq \Sigma$, report all positions $i \in [1..n-m+1]$ such that $\mathcal{P}$ set-parameterized matches $\mathcal{T}[i..i+m-1]$.
\end{problem}

\subsection{Our Contributions}

We present a linear time randomized algorithm for set parameterized matching which is novel in its solution approach.
This even dramatically improves the previous bipartite matching approaches for the special case of set parameterized comparison.
\begin{theorem}[Main Result]\label{thm:main}
Given a set-string text $\mathcal{T}[1..n]$ and set-string pattern $\mathcal{P}[1..m]$ over alphabet $\Sigma$, where $N = \sum_{i=1}^{n} |\mathcal{T}[i]|$ and $M = \sum_{j=1}^{m} |\mathcal{P}[j]|$, we can report all set-parameterized matches in time $O(N + M)$ with high probability using a Monte Carlo randomized algorithm with error probability at most $1/n$.
\end{theorem}

Key to our result is the introduction of a new structural definition we term the \emph{multiset offset representation}.
Unlike classical predecessor pointers which capture only local history, this representation preserves the global distribution of every character, ensuring that the necessary bijection properties hold even when multiple characters share identical positions.
To process this representation, which can be quadratically large in the worst case, we employ a novel three-layer hashing scheme:
\begin{enumerate}
\item \textbf{Layer 1}: Hash the offset set of each character to a single integer using Karp-Rabin fingerprinting, compressing the interaction history.
\item \textbf{Layer 2}: Hash the \emph{multiset} of these values at each position to a single integer ("mashing"), enabling efficient set-to-set comparisons that respect multiplicities.
\item \textbf{Layer 3}: Apply a rolling Karp-Rabin hash to the sequence of mashed values to enable linear time processing of sliding windows.
\end{enumerate}


As a warmup to the main result, we also demonstrate an efficient solution for the comparison problem:

\begin{theorem}[Set-String Comparison]\label{thm:comparison}
Given two set-strings $\mathcal{S}_1[1..m]$ and $\mathcal{S}_2[1..m]$ over alphabet $\Sigma$, where $M_{\mathcal{S}_1} = \sum_{i=1}^{m} |\mathcal{S}_1[i]|$ and $M_{\mathcal{S}_2} = \sum_{i=1}^{m} |\mathcal{S}_2[i]|$, we can determine whether they set-parameterized match in time $O(M_{\mathcal{S}_1} + M_{\mathcal{S}_2})$ with error probability at most $1/m$.
\end{theorem}

\subsection{Organization}

The remainder of the paper is organized as follows. Section~\ref{sec:prelim} defines our model and the necessary notation.
Section~\ref{sec:comparison} serves as a warmup, proving Theorem~\ref{thm:comparison} for comparing two set-strings. Section~\ref{sec:algorithm} details our main pattern matching algorithm, proving Theorem~\ref{thm:main}.
Finally, we analyze the correctness and complexity in Section~\ref{sec:analysis}.

\section{Model and Definitions}\label{sec:prelim}

\subsection{Notation and Basic Definitions}

Throughout this paper, we denote by $\Sigma$ a finite alphabet and by $\Sigma^*$ the set of all finite strings over $\Sigma$.
We assume that $\Sigma =\{1,2,\cdots, |\Sigma|\}$. If not, we sort $\Sigma$ and rename the alphabet characters.
For a string $S \in \Sigma^*$, we write $|S|$ for its size (length) and $S[i]$ for its $i$-th character, where $i \in [1..|S|]$.
We use the notation $[n] := \{1, 2, \ldots, n\}$ for positive integers $n$.
For a prime $p$, we denote the field $\mathbb{Z}_p := \{0, 1, \ldots, p-1\}$, the integers modulo $p$.

\subsection{Classical Parameterized Matching}

We begin by recalling the classical parameterized matching framework of Baker~\cite{baker1993theory,baker1997parameterized}, which forms the foundation of our work.
\begin{definition}[Parameterized Match~\cite{baker1993theory}]\label{def:pmatch}
Let $P[1..m], T[1..m] \in \Sigma^*$ be two strings of equal length over alphabet $\Sigma$.
We say that $P$ and $T$ \emph{parameterized match}, denoted $P \sim_p T$, if there exists a bijection $\pi: \Sigma \to \Sigma$ such that $T[i] = \pi(P[i])$ for all $i \in [1..m]$.
\end{definition}

\begin{remark}
The bijection $\pi$ in Definition~\ref{def:pmatch} must be defined on the entire alphabet $\Sigma$, though it need only be consistent on the characters actually appearing in $P$ and $T$.
Some works distinguish between "constant" and "parameter" alphabets $\Sigma_C \cup \Sigma_P$ where $\pi$ acts as the identity on $\Sigma_C$;
for simplicity, we consider the case where all characters are parameters (equivalently, $\Sigma_C = \emptyset$).
\end{remark}

\begin{problem}[Parameterized Pattern Matching]\label{prob:ppm}{\ }\\
\textbf{Input:} A text $T[1..n]$ and a pattern $P[1..m]$ over alphabet $\Sigma$.\\
\textbf{Output:} All positions $i \in [1..n-m+1]$ such that $P \sim_p T[i..i+m-1]$.
\end{problem}

Baker's key insight was to transform the string into a structural representation invariant under permutation.
This is the \emph{prev} transformation (or \emph{predecessor encoding}):

\begin{definition}[Prev Transformation~\cite{baker1993theory}]\label{def:prev}
For a string $S[1..n]$ over alphabet $\Sigma$, the \emph{prev} transformation $\text{prev}(S)$ is the integer sequence $\text{prev}(S)[1..n]$ where
\[
\text{prev}(S)[i] =
\begin{cases}
0 & \text{if } S[i] \text{ first occurs at position } i\\
i - j & \text{otherwise, where } j = \max\{k < i : S[k] = S[i]\}
\end{cases}
\]
\end{definition}

\begin{lemma}[Correctness of Prev~\cite{baker1993theory}]\label{lem:prev-correct}
For strings $P[1..m], T[1..m]$ over $\Sigma$, we have $P \sim_p T$ if and only if $\text{prev}(P) = \text{prev}(T)$ (equality of integer sequences).
\end{lemma}

\subsection{Set-Strings and Set Parameterized Matching}

We now formalize the set parameterized matching model, which generalizes the concepts above to sequences of sets.
\begin{definition}[Set-String]
A \emph{set-string} $\mathcal{S}$ of length $n$ over alphabet $\Sigma$ is a sequence $\mathcal{S}[1..n]$ where each $\mathcal{S}[i] \subseteq \Sigma$ is a finite set of characters.
The \emph{length} of $\mathcal{S}$ is $n$ and its \emph{size} is $|\mathcal{S}| = \sum_{j=1}^{n} |\mathcal{S}[j]|$.
We write $\mathcal{S}[i..j]$ for the subsequence $\mathcal{S}[i], \mathcal{S}[i+1], \ldots, \mathcal{S}[j]$.
\end{definition}

\begin{definition}[Set Parameterized Match]\label{def:set-pmatch}
Let $\mathcal{S}_1[1..m]$ and $\mathcal{S}_2[1..m]$ be two set-strings of equal length over alphabet $\Sigma$.
We say that $\mathcal{S}_1$ and $\mathcal{S}_2$ \emph{set-parameterized match}, denoted $\mathcal{S}_1 \approx_p \mathcal{S}_2$, if there exists a bijection $\pi: \Sigma \to \Sigma$ such that for all $i \in [1..m]$: $\pi(\mathcal{S}_1[i]) = \mathcal{S}_2[i]$

where $\pi(A) := \{\pi(a) : a \in A\}$ for any set $A \subseteq \Sigma$.
\end{definition}

\begin{remark}
Definition~\ref{def:set-pmatch} requires that the bijection $\pi$ simultaneously rename all characters in all sets consistently.
This generalizes Definition~\ref{def:pmatch}: if each set $\mathcal{P}[i]$ and $\mathcal{T}[i]$ is a singleton, we recover classical parameterized matching.
\end{remark}

\begin{problem}[Reformulation of Set Parameterized Matching]\label{prob:set-pm} {\ }\\
\textbf{Input:} A text $\mathcal{T}[1..n]$ and a pattern $\mathcal{P}[1..m]$, both set-strings over alphabet $\Sigma$.\\
\textbf{Output:} All positions $i \in [1..n-m+1]$ such that $\mathcal{P} \approx_p \mathcal{T}[i..i+m-1]$.
\end{problem}

For the special case where $n=m$ we call the problem the {\em Set Parameterized Comparison} as all that is required is to verify if the strings parameterize match each other.

\subsection{Offset Representation}

A naive extension of the prev transformation to set-strings—where one replaces every character in every set with its predecessor offset—fails because it destroys the continuity of character chains.
Consider the following counter-example:
Let set-string $\mathcal{S}_1 =\{\emptyset, \{b\}, \emptyset, \emptyset, \{a\},\emptyset, \emptyset, \{a,b\}, \emptyset,\emptyset,\emptyset, \emptyset, \{a\}\} $ and 
$\mathcal{S}_2 = \{\emptyset, \{c\}, \emptyset, \emptyset, \{d\},\emptyset, \emptyset, \{c,d\}, \emptyset,\emptyset,\emptyset, \emptyset, \{c\}\} $ 
If we simply replace each character with its prev value, the resulting sets of prevs match perfectly at every position:
$prev(\mathcal{S}_1) =prev(\mathcal{S}_2) =$ 
$$ \{\emptyset, \{0\}, \emptyset, \emptyset, \{0\},\emptyset, \emptyset, \{3,6\}, \emptyset,\emptyset,\emptyset, \emptyset, \{5\}\} $$

The sets of prevs are identical, yet there is no valid parameterization.
In $\mathcal{S}_1$, the character appearing at 13 ('a') is the one that started at 5. In $\mathcal{S}_2$, the character at 13 ('c') is the one that started at 2. The simple set of prevs $\{3, 6\}$ at position 8 obscures which prev belongs to which character chain.
To overcome this, we represent all appearances of a specific character ``together'' as a structural unit.
While a ``prev set'' (listing all predecessor offsets for a character) might capture this structure, it is not convenient for the rolling hash updates required by our algorithm.
Instead, we chose to represent this ``togetherness'' using \emph{offset sets}.
Since a character may appear in multiple sets, and a set contains multiple characters, a single "previous occurrence" is ill-defined.
We must instead track the \emph{set} of previous offsets. This motivates the following definition.
\begin{definition}[Offset Set]\label{def:offset-set}
Let $\mathcal{S}[1..n]$ be a set-string over $\Sigma$, and let $\sigma \in \Sigma$.
Suppose $\sigma$ appears in $\mathcal{S}$ at positions $i_1 < i_2 < \cdots < i_k$ (i.e., $\sigma \in \mathcal{S}[i_j]$ for all $j \in [1..k]$).
The \emph{offset set} of $\sigma$ in $\mathcal{S}$ is
\[
\text{OffSet}_{\mathcal{S}}(\sigma) := \{i_j - i_1 : j \in [1..k]\} \subseteq \{0, 1, 2, \ldots, n-1\}.
\]
If $\sigma$ does not appear in $\mathcal{S}$, we define $\text{OffSet}_{\mathcal{S}}(\sigma) := \emptyset$.
\end{definition}

\noindent Note that multisets are necessary as we may have two characters in one set with the same offset set.
We need to keep both of them as the multiplicity of characters in the compared set-strings must be equal.
With the offset set defined for each individual character, we now aggregate this information to represent the structure of the entire set-string.
At any given position $i$ in the string, multiple characters may be present.
To fully capture the parameterization constraints at this position, we must collect the offset sets of \emph{all} characters appearing there, forming the signature of that position, allowing us to compare the structural roles of characters across two different strings.
\begin{definition}[Multiset Offset Representation]
For a set-string $\mathcal{S}$ of length $m$, the offset representation at position $i$, denoted $\widehat{\mathcal{S}}[i]$, is the \emph{multiset} of offset sets for all characters appearing at $\mathcal{S}[i]$.
\[
\widehat{\mathcal{S}}[i] = \!\{\text{OffSet}_{\mathcal{S}}(\sigma) \mid \sigma \in \mathcal{S}[i]\!\}
\]

and 
\[
\widehat{\mathcal{S}} = \!\widehat{\mathcal{S}}[1]\widehat{\mathcal{S}}[2]\cdots \widehat{\mathcal{S}}[n]\!
\]

\end{definition}

In continuation with the former example: 

$\widehat{\mathcal{S}_1} =\{\emptyset, \{\{0,6\}\}, \emptyset, \emptyset, \{\{0,3,8\}\},\emptyset, \emptyset, \{\{0,6\},\{0,3,8\}\}, \emptyset,\emptyset,\emptyset, \emptyset, \{\{0,3,8\}\}\ \} $ and 

$\widehat{\mathcal{S}_2} = \{\emptyset, \{\{0,6,11\}\}, \emptyset, \emptyset, \{\{0,3\}\},\emptyset, \emptyset, \{\{0,6,11\},\{0,3\}\}, \emptyset,\emptyset,\emptyset, \emptyset, \{\{0,6,11\}\}\ \} $ 

\begin{observation}\label{obs:offset-size}
For a set-string $\mathcal{S}[1..n]$ over alphabet $\Sigma$, each offset set $\text{OffSet}_{\mathcal{S}}(\sigma)$ has size at most $n$.
Therefore, $\widehat{\mathcal{S}}[i]$ contains exactly $|\mathcal{S}[i]|$ offset sets, each of size at most $n$.
Letting $N = \sum_{i=1}^{n} |\mathcal{S}[i]|$ denote the total number of character occurrences, in the worst case, $|\widehat{\mathcal{S}}|$ (the total number of integers in the offset representation) is $O(nN)$.
\end{observation}

The following lemma, proof in appendix, shows that the Multiset Offset Representation provides a necessary and sufficient condition for set parameterized matching.
\begin{lemma}\label{lem:offset-correct}
Let $\mathcal{S}_1$ and $\mathcal{S}_2$ be set-strings of length $m$. Then $\mathcal{S}_1 \approx_p \mathcal{S}_2$ if and only if $\widehat{\mathcal{S}_1} = \widehat{\mathcal{S}_2}$ (where equality implies equality of the sequence of multisets).
\end{lemma}

\subsection{Challenges for Set Parameterized Comparison and Matching}

While Lemma~\ref{lem:offset-correct} successfully reduces set-parameterized comparison to checking equality of offset representations, it does not immediately yield an efficient algorithm for set parameterized comparison and, definitely, not for set parameterized matching,
due to two distinct hurdles:

\begin{enumerate}
\item \textbf{Size blowup:} By Observation~\ref{obs:offset-size}, the offset representation can be quadratically larger than the input.
Comparing these representations directly is prohibitively expensive.
\item \textbf{Dynamic offset sets:} For set parameterized \emph{matching}, we compare $\mathcal{P}$ against sliding windows $\mathcal{T}[i..i+m-1]$.
The offset sets $\text{OffSet}_{\mathcal{T}[i..i+m-1]}(\sigma)$ depend on the start index $i$ and must be recomputed for each position.
This contrasts with classical pattern matching, where the text "characters" remain static throughout the algorithm.
\end{enumerate}

Our algorithms addresses these challenges via a multi-layer \emph{hashing} scheme, detailed in Sections~\ref{sec:comparison} and~\ref{sec:algorithm}.
\subsection{Karp-Rabin Fingerprinting}

Our algorithms employ the Karp-Rabin fingerprinting technique~\cite{KarpR87}. We distinguish between two variants of the hash function, depending on whether the input is treated as a sequence (order matters) or a multiset (order irrelevant).
\begin{definition}[Sequence Hash]\label{def:kr-seq}
Let $p$ be a prime and $r \in \mathbb{Z}_p$.
For a sequence of integers $A = (a_0, a_1, \ldots, a_{k-1})$, the sequence hash treats the elements as coefficients of a polynomial:
\[
h_{p,r}^{seq}(A) := \sum_{i=0}^{k-1} a_i \cdot r^i \bmod p.
\]
\end{definition}

\begin{definition}[Set/Multiset Hash]\label{def:kr-set}
Let $S$ be a multiset of integers drawn from a universe $\{0, \dots, U-1\}$.
Let $c_v$ denote the multiplicity of value $v$ in $S$.
The set/multiset hash treats the values as exponents:
\[
h_{p,r}^{set}(S) := \sum_{v=0}^{U-1} c_v \cdot r^{v} \bmod p.
\]
\end{definition}

\begin{lemma}[Collision Probability~\cite{KarpR87}]\label{lem:kr-collision}
Let $p$ be a prime and $r \in \mathbb{Z}_p$ be chosen uniformly at random.
\begin{enumerate}
    \item \textbf{Sequence Collision:} Let $A$ and $B$ be two distinct sequences of length at most $\ell$.
    Then:
    \[ \Pr[h_{p,r}^{seq}(A) = h_{p,r}^{seq}(B)] \leq \frac{\ell-1}{p} \]
    \item \textbf{Multiset Collision:} Let $S_1$ and $S_2$ be two distinct multisets containing values from the universe $\{0, \dots, U-1\}$.
    Then:
    \[ \Pr[h_{p,r}^{set}(S_1) = h_{p,r}^{set}(S_2)] \leq \frac{U}{p} \]
\end{enumerate}
\end{lemma}

\begin{proof}
In both cases, the collision condition $h(X) = h(Y)$ is equivalent to $r$ being a root of the non-zero difference polynomial $Q(x) = h(X) - h(Y)$ over the field $\mathbb{Z}_p$.
By the fundamental theorem of algebra, the number of roots is bounded by the degree of $Q(x)$.
\begin{enumerate}
    \item For sequences, $Q(r) = \sum (a_i - b_i)r^i$.
    The degree is bounded by the maximum index, which is $\ell - 1$.
    \item For multisets, $Q(r) = \sum (c_{1,v} - c_{2,v})r^v$. The degree is bounded by the maximum value $v$ such that counts differ, which is at most $U$.
\end{enumerate}
Dividing the degree by the field size $p$ yields the probability bound.
\end{proof}

\subsection{Rolling Hashes}

In the classical Rabin-Karp pattern matching algorithm~\cite{KarpR87}, the rolling hash essentially computes a fingerprint for every window of length $m$ in the text, and compares it to the fingerprint of the pattern.
\begin{definition}[Rolling Hash]\label{def:rolling-hash}
For a sequence $S[1..n]$ where each $S[i] \in \mathbb{Z}_p$, and a window size $m \leq n$, define
\[
H_{p,r}^{(m)}(S, i) := \sum_{j=0}^{m-1} S[i+j] \cdot r^j \bmod p
\]
for $i \in [1..n-m+1]$.
This is the \emph{rolling hash} of the window $S[i..i+m-1]$.
\end{definition}

A key property of this scheme is the ability to efficiently update the hash value as the window slides from one position to the next.
\begin{observation}[Efficient Rolling Hash Update]\label{obs:rolling}
Given $H_{p,r}^{(m)}(S, i)$, we can compute $H_{p,r}^{(m)}(S, i+1)$ in $O(1)$ time via
\[
H_{p,r}^{(m)}(S, i+1) = \frac{H_{p,r}^{(m)}(S, i) - S[i]}{r} + S[i+m] \cdot r^{m-1} \bmod p.
\]
\end{observation}

Our algorithm, to appear later, builds upon this rolling concept.
However, adapting it to parameterized matching, and specifically set parameterized matching, introduces significant complexity: the effective value of a ``character'' (its offset) is relative to the window's start position and thus changes dynamically as the window rolls.
\section{Set Parameterized Comparison}\label{sec:comparison}

In this section, we prove Theorem~\ref{thm:comparison} by presenting an efficient algorithm for set parameterized comparison, i.e., determining whether two set-strings of equal length set-parameterized match.
This algorithm serves as a conceptual warmup for our main pattern matching algorithm and introduces the first two layers of our hashing scheme.
\subsection{Algorithm Overview}

The algorithm proceeds in three logical steps to progressively compress the set-string structure:
\begin{enumerate}
\item First, we compute the explicit offset sets for all characters in both set-strings.
\item Next, we hash each offset set to a single integer fingerprint using Karp-Rabin hashing (\textbf{Layer 1}).
\item Finally, we hash the set of fingerprints at each position to obtain a single value, allowing for linear comparisons (\textbf{Layer 2}).
\end{enumerate}

\subsection{Layer 1: Hashing Offset Sets}

We first address the computational cost of generating the offset sets. Recall that $\Sigma =\{1,2,\cdots, |\Sigma|\}$.
\begin{lemma}[Offset Set Computation]\label{lem:offset-computation}
Given a set-string $\mathcal{S}[1..m]$ over alphabet $\Sigma$ of size $M = \sum_{i=1}^{m} |\mathcal{S}[i]|$, we can compute $\text{OffSet}_{\mathcal{S}}(\sigma)$ for all $\sigma \in \Sigma$ in time $O(M)$ and space $O(M)$.
\end{lemma}

\begin{proof}
We perform a single pass over the set-string. We maintain a hash table $\textsc{FirstOcc}: \Sigma \to [1..m] \cup \{\bot\}$ to track the first occurrence of each character, and for each $\sigma \in \Sigma$, a list $\textsc{Offsets}[\sigma]$ initially empty.

\begin{algorithmic}[1]
\FOR{$\sigma \in \Sigma$}
    \STATE $\textsc{FirstOcc}[\sigma] \gets \bot$
    \STATE $\textsc{Offsets}[\sigma] \gets \emptyset$
\ENDFOR
\FOR{$i = 1$ to $m$}
    \FOR{$\sigma \in \mathcal{S}[i]$}
       \IF{$\textsc{FirstOcc}[\sigma] = \bot$}
           \STATE $\textsc{FirstOcc}[\sigma] \gets i$
           \STATE $\textsc{Offsets}[\sigma] \gets \textsc{Offsets}[\sigma] \cup \{0\}$
       \ELSE
           \STATE $\textsc{Offsets}[\sigma] \gets \textsc{Offsets}[\sigma] \cup \{i - \textsc{FirstOcc}[\sigma]\}$
       \ENDIF
   
 \ENDFOR
\ENDFOR
\end{algorithmic}

The initialization loop takes $O(|\Sigma|)$ time, which is $O(M)$.
The outer loop runs $m$ times, and the inner loop processes each character in each set.
In total, we process exactly $M = \sum_{i=1}^{m} |\mathcal{S}[i]|$ character-position pairs.
Each operation takes $O(1)$ expected time with hashing, giving a total time complexity of $O(M)$.

\end{proof}

Once the offset sets are computed, we compress them using the set hash function from Def.~\ref{def:kr-set}.
\begin{definition}[Layer 1 Hash]\label{def:layer1-hash}
Let $p_1$ be a prime and let $r_1 \in \mathbb{Z}_{p_1}$ be chosen uniformly at random.
For an offset set $O \subseteq \{0, 1, \ldots, m-1\}$, the Layer 1 hash function is:
\[
\phi_1(O) := h_{p_1, r_1}^{set}(O) = \sum_{x \in O} r_1^x \bmod p_1.
\]
\end{definition}

\emph{Collision Analysis:} By the definition, the Layer 1 hash is exactly the multiset hash $\phi_1(O) = h_{p_1, r_1}^{set}(O)$.
The elements of the offset sets are integers drawn from the universe $\{0, \dots, m-1\}$.
We apply Lemma~\ref{lem:kr-collision} (Multiset Collision case) with universe bound $U = m$, yielding for $O_1 \neq O_2$:
\[
\Pr[\phi_1(O_1) = \phi_1(O_2)] \leq \frac{m}{p_1}.
\]

\subsection{Layer 2: Hashing Sets of Fingerprints}

After applying Layer 1, each position $i$ in set-string $\mathcal{S}$ is represented by a multiset of fingerprints:
\[
\Phi_1(\mathcal{S}[i]) := \{\phi_1(\text{OffSet}_{\mathcal{S}}(\sigma)) : \sigma \in \mathcal{S}[i]\}.
\]

To compare positions efficiently, we must compare these multisets. We therefore apply a second layer of hashing to "mash" these multisets into single values.
Specifically, we use the multiset hash function from Definition~\ref{def:kr-set}.

\begin{definition}[Layer 2 Hash]\label{def:layer2-hash}
Let $p_2$ be a prime and let $r_2 \in \mathbb{Z}_{p_2}$ be chosen uniformly at random.
For a multiset $F$ of fingerprints from $\mathbb{Z}_{p_1}$, the Layer 2 hash function is:
\[
\phi_2(F) := h_{p_2, r_2}^{set}(F) = \sum_{f \in F} r_2^f \bmod p_2.
\]
\end{definition}

We call the resulting hash value the \emph{mashed} value of the position and the \emph{mashed representation} of a set-string $\mathcal{S}[1..m]$ is the sequence $\mathcal{M}(\mathcal{S}) = (M_1, M_2, \ldots, M_m)$ where $M_i := \phi_2(\Phi_1(\mathcal{S}[i])).$

\emph{Collision Analysis:} By the definition, Layer 2 is also the multiset hash $\phi_2(F) = h_{p_2, r_2}^{set}(F)$.
The elements of the multisets are Layer 1 fingerprints, which are integers drawn from the universe $\{0, \dots, p_1-1\}$.
Therefore, we can apply Lemma~\ref{lem:kr-collision} (Multiset Collision case) with universe bound $U = p_1$, yielding:

\[
\Pr_{r_2}[\phi_2(F_1) = \phi_2(F_2)] \leq \frac{p_1}{p_2}
\]

\subsection{Set Parameterized Comparison Algorithm}

The full procedure for set parameterized comparison, comparing two set-strings, consisting of applying the 2-level hash to each, is formalized in Algorithm~\ref{alg:compare} and appears in the appendix.

\begin{theorem}[Correctness of Comparison Algorithm]\label{thm:compare-correct}
Algorithm~\ref{alg:compare} correctly determines whether $\mathcal{S}_1 \approx_p \mathcal{S}_2$ with probability at least $1 - m^{-1}$.
\end{theorem}

\begin{proof}
\textbf{Completeness:}
Assume $\mathcal{S}_1 \approx_p \mathcal{S}_2$. By Lemma~\ref{lem:offset-correct}, we know $\widehat{\mathcal{S}_1} = \widehat{\mathcal{S}_2}$.
This implies that for every position $i$ and every $\sigma \in \mathcal{S}_1[i]$, there exists a corresponding $\tau \in \mathcal{S}_2[i]$ such that $\text{OffSet}_{\mathcal{S}_1}(\sigma) = \text{OffSet}_{\mathcal{S}_2}(\tau)$, and vice versa.
It follows deterministically that:
\begin{itemize}
\item $\Phi_1(\mathcal{S}_1[i]) = \Phi_1(\mathcal{S}_2[i])$ for all $i$ (as sets), and consequently,
\item $M^{\mathcal{S}_1}_i = M^{\mathcal{S}_2}_i$ for all $i$.
\end{itemize}
Thus, the algorithm correctly returns \texttt{true}.

\textbf{Soundness:}
Assume $\mathcal{S}_1 \not\approx_p \mathcal{S}_2$. By Lemma~\ref{lem:offset-correct}, $\widehat{\mathcal{S}_1} \neq \widehat{\mathcal{S}_2}$.
Thus, there must exist some position $i^*$ where the sets of offset sets differ: $\widehat{\mathcal{S}_1}[i^*] \neq \widehat{\mathcal{S}_2}[i^*]$.
Let $\mathcal{E}_1$ be the event that "there exist distinct offset sets $O_1, O_2$ with $\phi_1(O_1) = \phi_1(O_2)$."
By choosing $p_1 = \Theta\left( {m^2|\Sigma|^2} \right)$ and taking the union bound over all $\binom{|\Sigma|}{2} \leq |\Sigma|^2$ pairs of distinct characters, we bound the probability of this collision:
\[
\Pr[\mathcal{E}_1] \leq |\Sigma|^2 \cdot \frac{m}{p_1} = O\left(\frac{m|\Sigma|^2}{m^2 |\Sigma|^2}\right) = O\left(\frac{1}{m}\right).
\]

Conditioned on $\neg \mathcal{E}_1$, the hash function $\phi_1$ is injective on the offset sets appearing in $\mathcal{S}_1$ and $\mathcal{S}_2$.
Since $\widehat{\mathcal{S}_1}[i^*] \neq \widehat{\mathcal{S}_2}[i^*]$, it must be that $\Phi_1(\mathcal{S}_1[i^*]) \neq \Phi_1(\mathcal{S}_2[i^*])$ as sets of fingerprints.
Next, let $\mathcal{E}_2$ be the event that "there exist distinct sets $F_1, F_2 \subseteq \mathbb{Z}_{p_1}$ with $\phi_2(F_1) = \phi_2(F_2)$."
Choose $p_2 = \Theta\left( {m^4|\Sigma|^2} \right)$. We are comparing $m$ multisets (one for each position $i \in [1..m]$).
For any specific position, the collision probability is bounded by $p_1/p_2$ (by Lemma~\ref{lem:kr-collision} with universe $p_1$).
Applying the union bound over all $m$ positions:
\[
\Pr[\mathcal{E}_2 \mid \neg \mathcal{E}_1] \leq m \cdot \frac{p_1}{p_2} =  O\left( \frac{m^3|\Sigma|^2}{m^4|\Sigma|^{2}}\right) = O\left(\frac{1}{m}\right).
\]

Conditioned on $\neg \mathcal{E}_1 \wedge \neg \mathcal{E}_2$, we are guaranteed that $M^{\mathcal{S}_1}_{i^*} \neq M^{\mathcal{S}_2}_{i^*}$, and so the algorithm returns \texttt{false}.
The overall error probability is bounded by the sum of probabilities of these collision events:
\[
\Pr[\mathcal{E}_1 \vee \mathcal{E}_2] \leq \Pr[\mathcal{E}_1] + \Pr[\mathcal{E}_2 \mid \neg \mathcal{E}_1] = O\left(\frac{1}{m}\right).
\]
\end{proof}

This concludes the proof of Theorem~\ref{thm:comparison}, with a summary appearing in the index.

\section{Set Parameterized Matching Algorithm}\label{sec:algorithm}

We now present our main algorithm for Problem~\ref{prob:set-pm}, the problem of Set Parameterized Matching, proving Theorem~\ref{thm:main}.
\subsection{The Rolling Hash Challenge}

Recall from Section~\ref{sec:prelim} (Definition~\ref{def:rolling-hash}) that the Rabin-Karp rolling hash allows for efficient $O(1)$ updates when sliding a window over a fixed sequence of integers.
In our context, the natural candidate for such a sequence is the sequence of Layer 2 mashed values, $M^{\mathcal{T}}_1, M^{\mathcal{T}}_2, \dots, M^{\mathcal{T}}_n$.
A direct application of the standard rolling hash for a window of size $m$ starting at text position $i$ would be:
\[
H = \sum_{j=0}^{m-1} M^{\mathcal{T}}_{i+j} \cdot r_3^j \bmod p_3
\]
where $p_3$ is some prime, $r_3 \in [p_3]$, and $M^{\mathcal{T}}_{k}$ represents the Layer 2 hash of the offset sets at text position $k$.

\subsubsection*{The Dynamic Offset Problem}
However, this standard approach fails because the values $M^{\mathcal{T}}_k$ are \emph{not static}.
In set parameterized matching, the offset set $\text{OffSet}_{\mathcal{T}[i..i+m-1]}(\sigma)$ is defined relative to the \emph{first occurrence} of $\sigma$ inside the current window.
When the window slides from $i$ to $i+1$, if a character $\sigma$ is dropped from position $i$, the "first occurrence" of $\sigma$ changes to its next appearance in the window (say, at relative distance $\delta$).
Consequently, the offsets of \emph{all} remaining occurrences of $\sigma$ in the window must be decreased by $\delta$.
This changes the Layer 1 hashes ($\psi$) and, crucially, changes the Layer 2 mashed values ($M^{\mathcal{T}}$) for every position where $\sigma$ appears.
Updating the rolling hash $H$ would thus require recomputing up to $O(m)$ terms in the summation, which is prohibitively slow.
\subsection{Character-Centric Rolling Hash}

To handle these dynamic updates efficiently, we restructure the hashing scheme.
Instead of viewing $H$ as a sum over window positions (which is sensitive to shifting offsets), we swap the order of summation to view it as a sum over alphabet characters.
Recall that for the window starting at location $i$, $M^{\mathcal{T}}_{i+j} = \sum_{\sigma \in \mathcal{T}[i+j]} r_2^{\psi^{(i)}[\sigma]}$.
Substituting this into the standard rolling hash definition:
\[
H = \sum_{j=0}^{m-1} r_3^j \left( \sum_{\sigma \in \mathcal{T}[i+j]} r_2^{\psi^{(i)}[\sigma]} \mod p_2 \right) \mod p_3
\]
By swapping the summation order, we group terms by character rather than by position:
\[
H = \sum_{\sigma \in \Sigma} \left( r_2^{\psi^{(i)}[\sigma]} \bmod p_2 \right) \left( \sum_{j : \sigma \in \mathcal{T}[i+j]} r_3^j \right) \bmod p_3
\]
This rearrangement separates the \emph{structural shape} of the encoding of $\sigma$ (encoded by $r_2^{\psi^{(i)}[\sigma]}$) from the \emph{geometric positions} of the encoding of $\sigma$ (encoded by the inner sum over $r_3$).
This separation allows us to update the two components independently.

This motivates our formal definition of the components for the character-centric hash.
We first isolate the ``geometric'' component—the inner sum representing the positions of a character.
\begin{definition}[Character Position Fingerprint]\label{def:pos-fingerprint}
Let $p_2$ be a prime and let $r_2 \in \mathbb{Z}_{p_2}$ be chosen uniformly at random. For a character $\sigma$ in the window starting at text position $i$, let its \emph{position fingerprint} be the sum of the position powers for all its occurrences:
\[
P_i(\sigma) := \sum_{k \in \text{Occurrences of } \sigma \text{ in window}} r_3^k \bmod p_3
\]
where $k \in \{0, \dots, m-1\}$ is the relative index of the occurrence within the window.
\end{definition}

With the geometric component isolated, the total window hash can be expressed as a weighted sum of these fingerprints, where the weights are determined by the structural shape (offset sets) of each character.
\begin{definition}[Layer 3 Hash]\label{def:layer3-hash}
The character-centric Layer 3 hash for the window at text position $i$ is:
\[
H^{\mathcal{T}}_i := \sum_{\sigma \in \Sigma} \left( r_2^{\psi^{(i)}[\sigma]} \bmod p_2 \right) \cdot P_i(\sigma) \bmod p_3
\]
where $\psi^{(i)}[\sigma]$ is the Layer 1 hash of the offset set of $\sigma$ in window $i$.
\end{definition}

\subsection{Efficient Transition Logic}\label{sec:incremental-update}

To achieve our time complexity goals, we must compute $H^{\mathcal{T}}_{i+1}$ from $H^{\mathcal{T}}_i$ in time proportional only to the number of characters entering or leaving the window, independent of the alphabet size $|\Sigma|$.
When sliding the window from $i$ to $i+1$, two distinct types of updates occur:
\begin{enumerate}
    \item \textbf{Geometric Shift:} Every character remaining in the window shifts left by one position.
This implies that for \emph{every} $\sigma \in \Sigma$, the position fingerprint changes: $P_{i+1}(\sigma) = P_i(\sigma) \cdot r_3^{-1}$.
\item \textbf{Structural Change:} For characters in $\mathcal{T}[i]$ (leaving) and $\mathcal{T}[i+m]$ (entering), the offset structure $\psi^{(i)}[\sigma]$ may change to a new value $\psi^{(i+1)}[\sigma]$.
\end{enumerate}

A naive update is impossible because we cannot iterate over all $\sigma \in \Sigma$ to apply the geometric shift $r_3^{-1}$, nor can we simply multiply the total hash $H^{\mathcal{T}}_i$ by $r_3^{-1}$ (since the structural weights $r_2^{\psi[\sigma]}$ for active characters would be incorrectly shifted).
Instead, we decompose the total hash into a \emph{Stable Component} (characters with unchanged structure) and an \emph{Active Component} (characters appearing at the window boundaries).
We process the transition in three steps: \emph{Isolate}, \emph{Shift}, and \emph{Reintegrate}.
Let $S_{active} = \mathcal{T}[i] \cup \mathcal{T}[i+m]$ be the set of characters undergoing structural updates.
\subsubsection*{Step 1: Isolate (Remove Active Terms)}
First, we remove the contributions of all "active" characters from the current hash to isolate the "stable" portion.
For any stable character $\sigma \notin S_{active}$, the structural hash $\psi$ remains invariant ($\psi^{(i)}[\sigma] = \psi^{(i+1)}[\sigma]$).
\[
H_{stable} = H^{\mathcal{T}}_i - \sum_{\sigma \in S_{active}} \left( r_2^{\psi^{(i)}[\sigma]} \bmod p_2 \right) \cdot P_i(\sigma) \bmod p_3
\]
This step requires $O(|S_{active}|)$ operations.
\subsubsection*{Step 2: Global Shift (Update Stable Terms)}
The remaining hash $H_{stable}$ consists essentially of a sum of terms $C_\sigma \cdot P_i(\sigma)$ where $C_\sigma$ is constant.
The only required update is the geometric shift of positions ($k \to k-1$).
By the linearity of the sum, we can apply this shift globally:
\[
H'_{stable} = H_{stable} \cdot r_3^{-1} \bmod p_3
\]
This single multiplication updates the hash for all non-active characters in $O(1)$ time.
\textbf{Lazy Fingerprint Maintenance:} To support this global shift without updating individual $P(\sigma)$ values in memory, we maintain a global rolling multiplier $R_{glob}$.
The stored value for a fingerprint is related to the actual value by $P_{actual}(\sigma) = P_{stored}(\sigma) \cdot R_{glob}$.
In this step, we simply update $R_{glob} \leftarrow R_{glob} \cdot r_3^{-1}$.
\subsubsection*{Step 3: Reintegrate (Add Updated Active Terms)}
Finally, we compute the new states for the active characters and add them back to the shifted stable hash.
For each $\sigma \in S_{active}$:
\begin{itemize}
    \item \textbf{Update Position:} We update $P_{stored}(\sigma)$ to reflect any removals (at relative index 0) or insertions (at relative index $m-1$).
Note that we must multiply any added terms by $R_{glob}^{-1}$ to align with the current global frame.
\item \textbf{Update Offset Structure:}
    \begin{itemize}
        \item \emph{Leaving ($\sigma \in \mathcal{T}[i]$):} We must remove the first occurrence.
Let $\delta$ be the distance to the \emph{next} occurrence of $\sigma$ in the window.
The structure updates by shifting the exponents:
        \[ \psi^{(i+1)}[\sigma] \leftarrow (\psi^{(i)}[\sigma] - r_1^0) \cdot r_1^{-\delta} \bmod p_1 \]
        (If $\sigma$ does not appear again, $\psi^{(i+1)}[\sigma]$ resets to 0).
\item \emph{Entering ($\sigma \in \mathcal{T}[i+m]$):} We must add the new occurrence at the end of the window.
Let $k$ be the index of the \emph{first occurrence} of $\sigma$ in the new window (if $\sigma$ is new to the window, $k=m-1$).
The new relative offset to add is $\delta_{new} = (m-1) - k.$
\[ \psi^{(i+1)}[\sigma] \leftarrow \psi^{(i+1)}[\sigma] + r_1^{\delta_{new}} \bmod p_1 \]
    \end{itemize}
\end{itemize}

We then add the new contribution to the total hash:
\[
H^{\mathcal{T}}_{i+1} = H'_{stable} + \sum_{\sigma \in S_{active}} \left( r_2^{\psi^{(i+1)}[\sigma]} \bmod p_2 \right) \cdot (P_{stored}(\sigma) \cdot R_{glob}) \bmod p_3
\]
This step also takes time $O(|S_{active}|)$.

Note:

(1) The definition of offset sets comes into its main play here, that is, the update of $\psi^{(i+1)}[\sigma]$ is computable in $O(1)$ time because of the structure of the offset set.

(2) In the computation of the entering characters, we assume knowledge of the \emph{first occurrence} of $\sigma$.
We do this by saving linked lists of the locations of each $\sigma$ and traverse them while shifting, changing pointers to current \emph{first occurrence} in each list when characters leave the window.
For the sake of brevity, we keep the details to a minimum.
\begin{theorem}[Transition Complexity]
The transition from $H^{\mathcal{T}}_i$ to $H^{\mathcal{T}}_{i+1}$ can be computed in time $O(1 + |\mathcal{T}[i]| + |\mathcal{T}[i+m]|)$.
\end{theorem}

\begin{proof}
Step 2 takes $O(1)$ time. Steps 1 and 3 involve iterating only over the characters present at the leaving position $i$ and the entering position $i+m$.
All internal lookups (offset lists, hash maps) take $O(1)$ time.
Thus, the total time is proportional to the number of character occurrences added or removed.
\end{proof}

The full algorithm, depicting the above, is detailed in Algorithm~\ref{alg:main} in the appendix.

\section{Analysis}\label{sec:analysis}

\subsection{Correctness Analysis}

We now prove that the hashing scheme correctly identifies matches with high probability.
\begin{theorem}[Correctness]\label{thm:main-correct}
Algorithm~\ref{alg:main} reports all valid set-parameterized matches and reports false positives with probability at most $n^{-1}$.
\end{theorem}

\begin{proof}
The matching algorithm effectively performs a set-string comparison at each of the $n-m+1$ text windows.
We analyze the error probability by extending the proof of Theorem~\ref{thm:compare-correct} by choosing primes:
$p_1 = \Theta(n^2|\Sigma|^2 m)$, $p_2 = \Theta(n^4|\Sigma|^2 m^2)$, and $p_3 = \Theta(n^2m)$.

\begin{enumerate}
    \item \textbf{Layers 1 \& 2 (Structural Collision):}
    By the analysis in Theorem~\ref{thm:compare-correct}, the probability of a collision in the structural hashing (Layers 1 and 2) for any single window is bounded by:
    \[
    \Pr[\mathcal{E}_{struct}] \leq \Pr[\mathcal{E}_1] + \Pr[\mathcal{E}_2] \le \frac{m|\Sigma|^2}{p_1} + \frac{m p_1}{p_2} = O\left(\frac{1}{n^2}\right) + O\left(\frac{1}{n^2}\right) = O\left(\frac{1}{n^2}\right).
    \]

    \item \textbf{Layer 3 (Rolling Hash Collision):}
    The rolling hash $H$ is a polynomial in $r_3$ of degree at most $m$. For any single window, using Lemma~\ref{lem:kr-collision}, the collision probability is:
    \[
    \Pr[\mathcal{E}_{roll}] \leq \frac{m}{p_3} = O\left(\frac{m}{n^2m}\right) = O\left(\frac{1}{n^2}\right).
    \]
\end{enumerate}

\noindent \textbf{Total Error:}
Applying the union bound over all $O(n)$ sliding windows, the total probability of a false positive is:
\[
\Pr[\text{Error}] \leq n \cdot \left( \Pr[\mathcal{E}_{struct}] + \Pr[\mathcal{E}_{roll}] \right)
= n \cdot O\left(\frac{1}{n^2}\right)
= O\left(\frac{1}{n}\right).
\] 
\end{proof}

\subsection{Complexity Analysis}

Let $N = \sum_{i=1}^n |\mathcal{T}[i]|$ and $M = \sum_{j=1}^m |\mathcal{P}[j]|$. The following theorem summarizes the result. Details appear in the appendix.
\begin{theorem}[Complexity]\label{thm:main-complexity}
Algorithm~\ref{alg:main} runs in time $O(N + M)$ and uses space $O(N + M)$.
\end{theorem}

\bibliographystyle{plainurl}
\bibliography{references}

\appendix
\newpage
\section{Proof of Correctness for the Offset Sets}

The following is the proof of Lemma~\ref{lem:offset-correct}.

\begin{proof}
\textbf{Direction ($\Rightarrow$):}
Suppose $\mathcal{S}_1 \approx_p \mathcal{S}_2$. There exists a bijection $\pi: \Sigma \to \Sigma$ such that $\pi(\mathcal{S}_1[i]) = \mathcal{S}_2[i]$ for all $i$.
This implies that for every $\sigma$, the absolute positions of $\sigma$ in $\mathcal{S}_1$ are identical to the absolute positions of $\pi(\sigma)$ in $\mathcal{S}_2$.
Consequently, they share the same start position and the same offset set.
Thus, the multiset of offset sets at every position remains invariant, and $\widehat{\mathcal{S}_1} = \widehat{\mathcal{S}_2}$.

\textbf{Direction ($\Leftarrow$):}
Suppose $\widehat{\mathcal{S}_1} = \widehat{\mathcal{S}_2}$.
We must construct a valid bijection $\pi$.

A valid bijection must map a character $\sigma \in \mathcal{S}_1$ to a character $\tau \in \mathcal{S}_2$ if and only if they are structurally identical.
Structural identity requires two conditions:
\begin{enumerate}
    \item They have the same \textbf{Offset Set} $O$ (relative spacing).
    \item They have the same \textbf{Start Position} $s$ (absolute placement).
\end{enumerate}

Let $N^{\mathcal{S}}(s, O)$ denote the number of distinct characters in string $\mathcal{S}$ that start at position $s$ and have offset set $O$.
To prove a bijection exists, it is sufficient to show that for all positions $s \in [1..m]$ and all offset sets $O$:
\[
N^{\mathcal{S}_1}(s, O) = N^{\mathcal{S}_2}(s, O)
\]
If this equality holds, we can map the $k$ characters of type $(s, O)$ in $\mathcal{S}_1$ arbitrarily to the $k$ characters of type $(s, O)$ in $\mathcal{S}_2$.
\textbf{Computing $N(s, O)$ from the Representation:}
We show that $N^{\mathcal{S}}(s, O)$ is uniquely determined by the sequence of multisets $\widehat{\mathcal{S}}$.
For any position $i$ and offset set $O$, let $Count^{\mathcal{S}}(i, O)$ be the multiplicity of $O$ in the multiset $\widehat{\mathcal{S}}[i]$.
A character $\sigma$ contributes to the count of $O$ at position $i$ if:
\begin{enumerate}
    \item $\text{OffSet}(\sigma) = O$.
    \item $\sigma$ appears at $i$. This happens if $\sigma$ started at some position $s \le i$ such that $(i-s) \in O$.
\end{enumerate}

We can express the total count at position $i$ as the sum of characters starting \emph{now} ($s=i$) and characters that started \emph{earlier} ($s < i$) and are re-appearing:
\[
Count^{\mathcal{S}}(i, O) = \underbrace{N^{\mathcal{S}}(i, O)}_{\text{Starts at } i} + \sum_{\delta \in O, \delta > 0} \underbrace{N^{\mathcal{S}}(i-\delta, O)}_{\text{Started at } i-\delta}
\]

Rearranging to solve for the current starts:
\[
N^{\mathcal{S}}(i, O) = Count^{\mathcal{S}}(i, O) - \sum_{\delta \in O, \delta > 0} N^{\mathcal{S}}(i-\delta, O)
\]

\textbf{Inductive Argument:}
We prove $N^{\mathcal{S}_1}(i, O) = N^{\mathcal{S}_2}(i, O)$ for all $i$ by strong induction.
\begin{itemize}
    \item \textbf{Base Case ($i=1$):} There are no offsets $\delta > 0$ relative to 1. Thus $N^{\mathcal{S}}(1, O) = Count^{\mathcal{S}}(1, O)$.
    Since $\widehat{\mathcal{S}_1} = \widehat{\mathcal{S}_2}$, the counts are identical, so $N^{\mathcal{S}_1}(1, O) = N^{\mathcal{S}_2}(1, O)$.
    \item \textbf{Inductive Step:} Assume $N^{\mathcal{S}_1}(j, O) = N^{\mathcal{S}_2}(j, O)$ for all $j < i$.
    In the equation for $N(i, O)$:
    \begin{itemize}
        \item $Count(i, O)$ is identical because $\widehat{\mathcal{S}_1}[i] = \widehat{\mathcal{S}_2}[i]$.
        \item The summation terms $N(i-\delta, O)$ involve indices strictly less than $i$.
        By the induction hypothesis, these terms are identical for $\mathcal{S}_1$ and $\mathcal{S}_2$.
    \end{itemize}
    Therefore, $N^{\mathcal{S}_1}(i, O) = N^{\mathcal{S}_2}(i, O)$.
\end{itemize}

Since the number of characters for every specific (Start Position, Offset Set) pair is identical, we can construct a perfect bijection $\pi$.
\end{proof}

\section{Proof of Theorem~\ref{thm:comparison}}

\begin{theorem}[Complexity of Comparison Algorithm]\label{thm:compare-complexity}
Let $M_{\mathcal{S}_1} = \sum_{i=1}^{m} |\mathcal{S}_1[i]|$ and $M_{\mathcal{S}_2} = \sum_{i=1}^{m} |\mathcal{S}_2[i]|$ denote the total size (number of character occurrences) of the two set-strings.
Algorithm~\ref{alg:compare} runs in time $O(M_{\mathcal{S}_1} + M_{\mathcal{S}_2})$ and uses space $O(M_{\mathcal{S}_1} + M_{\mathcal{S}_2})$.
\end{theorem}

\begin{proof}
\textbf{Time Analysis:}
\begin{itemize}
\item \textbf{Initialization:} Selecting primes and initializing tables takes $O(|\Sigma|)$ time.
\item \textbf{Offset Set Computation:} By Lemma~\ref{lem:offset-computation}, computing the offset sets requires iterating through every character occurrence in both strings exactly once.
The time required is $O(M_{\mathcal{S}_1} + M_{\mathcal{S}_2})$.
\item \textbf{Layer 1 Hashing:} We calculate $\phi_1(O)$ for every character $\sigma \in \Sigma$.
The cost of computing the hash for a specific offset set $O$ is proportional to its size $|O|$.
Since the sum of the sizes of all offset sets is exactly the total number of character occurrences, the total time is $\sum_{\sigma} |O_{\sigma}|
= O(M_{\mathcal{S}_1} + M_{\mathcal{S}_2})$.
\item \textbf{Layer 2 Hashing:} We calculate $\phi_2$ for every position $i \in [1..m]$.
The cost at position $i$ is proportional to the number of characters in the set at that position, $|\mathcal{S}_1[i]|$ (or $|\mathcal{S}_2[i]|$).
Summing over all positions $i$, the total time is $\sum_{i=1}^{m} |\mathcal{S}_1[i]| + \sum_{i=1}^{m} |\mathcal{S}_2[i]| = O(M_{\mathcal{S}_1} + M_{\mathcal{S}_2})$.
\item \textbf{Comparison:} Comparing the two mashed sequences takes $O(m)$ time, which is dominated by $O(M_{\mathcal{S}_1} + M_{\mathcal{S}_2})$ since $m \le M$.
\end{itemize}
Since $|\Sigma| = O(M_{\mathcal{S}_1} + M_{\mathcal{S}_2})$ the total running time is $O(M_{\mathcal{S}_1} + M_{\mathcal{S}_2})$.
\textbf{Space Analysis:}
Storing the offset sets and the intermediate hash values requires space proportional to the total number of elements in the sets, which is $O(M_{\mathcal{S}_1} + M_{\mathcal{S}_2})$.
\end{proof}

\section{Proof of Complexity Analysis}

\begin{proof}
\textbf{Time Complexity:}
\begin{enumerate}
    \item \textbf{Preprocessing:} Building the occurrence lists $\mathcal{L}$ takes $O(N)$.
Computing the initial pattern hash takes $O(M)$.
    \item \textbf{Sliding Window Scan:} The loop runs for $n-m$ steps.
Inside the loop, the "Global Shift" (Step 2) takes $O(1)$ time.
The "Isolate" (Step 1) and "Reintegrate" (Step 3) phases iterate only over the set of active characters $S_{active} = \mathcal{T}[i] \cup \mathcal{T}[i+m]$.
Let $Cost(i)$ be the number of operations at step $i$. We have $Cost(i) = O(1 + |\mathcal{T}[i]| + |\mathcal{T}[i+m]|)$.
Summing over all windows:
    \[
    \text{Total Time} = \sum_{i=1}^{n-m} Cost(i) = \sum_{i=1}^{n-m} \left( O(1) + |\mathcal{T}[i]| + |\mathcal{T}[i+m]| \right)
    \]
    Observe that each character occurrence in the text at position $k$ appears in the set $\mathcal{T}[i]$ exactly once (when it leaves the window) and in the set $\mathcal{T}[i+m]$ exactly once (when it enters the window).
Therefore:
    \[
    \sum_{i=1}^{n-m} \left( |\mathcal{T}[i]| + |\mathcal{T}[i+m]| \right) \leq 2 \sum_{k=1}^{n} |\mathcal{T}[k]|
= 2N
    \]
    \item \textbf{Hash Updates:} All internal hash updates (multiplications, additions) and list lookups (finding next occurrence $\delta$) take $O(1)$ time.
\end{enumerate}
Thus, the total time complexity is $O(N + M)$.

\textbf{Space Complexity:}
The occurrence lists $\mathcal{L}$ store every character position, requiring $O(N)$ space.
The hash maps for $P_{stored}$ and $\psi$ store one entry per unique character currently in the window.
Since the number of unique characters cannot exceed the number of character occurrences, this is bounded by $O(N)$.
The pattern storage takes $O(M)$. Thus, the total space is $O(N + M)$.
\end{proof}

\section{Algorithms}

\begin{algorithm}
\caption{$\textsc{CompareSetStrings}(\mathcal{S}_1, \mathcal{S}_2)$}
\label{alg:compare}
\begin{algorithmic}[1]
\REQUIRE Set-strings $\mathcal{S}_1[1..m]$ and $\mathcal{S}_2[1..m]$ over alphabet $\Sigma$
\ENSURE \texttt{true} if $\mathcal{S}_1 \approx_p \mathcal{S}_2$, \texttt{false} otherwise (with high probability)
\STATE Choose prime $p_1 = \Theta((m^2|\Sigma|)^2)$ and random $r_1 \in \mathbb{Z}_{p_1}$
\STATE Choose prime $p_2 = \Theta((m^4|\Sigma|)^{2})$ and random $r_2 \in \mathbb{Z}_{p_2}$
\STATE Compute $\text{OffSet}_{\mathcal{S}_1}(\sigma)$ for all $\sigma \in \Sigma$ \COMMENT{Lemma~\ref{lem:offset-computation}}
\STATE Compute $\text{OffSet}_{\mathcal{S}_2}(\sigma)$ for all $\sigma \in \Sigma$
\FOR{$\sigma \in \Sigma$}
    \STATE $\psi_{\mathcal{S}_1}[\sigma] \gets \phi_1(\text{OffSet}_{\mathcal{S}_1}(\sigma))$ \COMMENT{Layer 1 hash}
    \STATE $\psi_{\mathcal{S}_2}[\sigma] \gets \phi_1(\text{OffSet}_{\mathcal{S}_2}(\sigma))$
\ENDFOR
\FOR{$i = 1$ to $m$}
    \STATE $M^{\mathcal{S}_1}_i \gets \phi_2(\{\psi_{\mathcal{S}_1}[\sigma] : \sigma \in \mathcal{S}_1[i]\})$ \COMMENT{Layer 2 hash}
    \STATE $M^{\mathcal{S}_2}_i \gets \phi_2(\{\psi_{\mathcal{S}_2}[\sigma] : \sigma 
\in \mathcal{S}_2[i]\})$
\ENDFOR
\RETURN $(M^{\mathcal{S}_1}_1, \ldots, M^{\mathcal{S}_1}_m) = (M^{\mathcal{S}_2}_1, \ldots, M^{\mathcal{S}_2}_m)$
\end{algorithmic}
\end{algorithm}

\begin{algorithm}
\caption{$\textsc{SetParameterizedMatching}(\mathcal{P}, \mathcal{T})$}
\label{alg:main}
\begin{algorithmic}[1]

\STATE \textbf{Preprocessing:}
\STATE Build occurrence lists $\mathcal{L}$ for text.
\STATE Compute pattern hash $H^{\mathcal{P}}$.
\STATE Compute $H^{\mathcal{T}}_1$ explicitly. Initialize $P_{stored}(\sigma)$ and $\psi[\sigma]$.
\STATE Set global multiplier $R_{glob} \gets 1$.
\STATE $Candidates \gets \{1\}$ if $H^{\mathcal{T}}_1 = H^{\mathcal{P}}$.

\STATE \textbf{Sliding Window:}
\FOR{$i = 1$ to $n - m$}
    \STATE $H_{stable} \gets H^{\mathcal{T}}_i$
    \STATE $S_{active} \gets \mathcal{T}[i] \cup \mathcal{T}[i+m]$

    \STATE \textbf{1.
Isolate (Remove Old Contributions):}
    \FOR{$\sigma \in S_{active}$}
        \STATE $P_{curr} \gets P_{stored}(\sigma) \cdot R_{glob} \bmod p_3$
        \STATE $H_{stable} \gets H_{stable} - (r_2^{\psi^{(i)}[\sigma]} \bmod p_2) \cdot P_{curr} \bmod p_3$
    \ENDFOR

    \STATE \textbf{2.
Global Shift:}
    \STATE $H_{stable} \gets H_{stable} \cdot r_3^{-1} \bmod p_3$
    \STATE $R_{glob} \gets R_{glob} \cdot r_3^{-1} \bmod p_3$

    \STATE \textbf{3.
Reintegrate (Update \& Add New Contributions):}
    \FOR{$\sigma \in S_{active}$}
        \STATE \textit{// Update Position Fingerprint}
        \IF{$\sigma \in \mathcal{T}[i]$} 
            \STATE $P_{stored}(\sigma) \gets P_{stored}(\sigma) - (r_3^0 \cdot R_{glob}^{-1})$ \COMMENT{Remove leaving position}
        \ENDIF
        \IF{$\sigma \in \mathcal{T}[i+m]$}
            \STATE $P_{stored}(\sigma) \gets P_{stored}(\sigma) + (r_3^{m-1} \cdot R_{glob}^{-1})$ \COMMENT{Add entering position}
    
    \ENDIF
        
        \STATE \textit{// Update Offset Structure (Layer 1)}
        \IF{$\sigma \in \mathcal{T}[i]$}
            \STATE Find next occurrence distance $\delta$ using $\mathcal{L}$.
\IF{$\delta \le m$}
                \STATE $\psi[\sigma] \gets (\psi[\sigma] - r_1^0) \cdot r_1^{-\delta} \bmod p_1$
            \ELSE
                \STATE $\psi[\sigma] \gets 0$ \COMMENT{Reset if no more occurrences}
            \ENDIF
        \ENDIF
        \IF{$\sigma \in \mathcal{T}[i+m]$}
       
     \STATE Determine relative offset $\delta_{new}$ of new occurrence.
\STATE $\psi[\sigma] \gets \psi[\sigma] + r_1^{\delta_{new}} \bmod p_1$
        \ENDIF

        \STATE $P_{curr} \gets P_{stored}(\sigma) \cdot R_{glob} \bmod p_3$
        \STATE $H_{stable} \gets H_{stable} + (r_2^{\psi[\sigma]} \bmod p_2) \cdot P_{curr} \bmod p_3$
    \ENDFOR

    \STATE $H^{\mathcal{T}}_{i+1} \gets H_{stable}$
    \IF{$H^{\mathcal{T}}_{i+1} = H^{\mathcal{P}}$}
       \STATE $\textit{Candidates} \gets \textit{Candidates} \cup \{i+1\}$
    \ENDIF
\ENDFOR

\RETURN $\textit{Candidates}$
\end{algorithmic}
\end{algorithm}

\end{document}